\documentclass[sigconf,nonacm]{acmart}
\usepackage[utf8]{inputenc}
\usepackage[font=small,skip=1pt]{caption}
\usepackage{etoolbox}
\makeatletter

\title{Simpson's Paradox in Recommender Fairness: Reconciling differences between per-user and aggregated evaluations}
\author[F. Prost, B. Packer, J.Chen, L. Wei, P. Kremp, N. Blumm, S. Wang, T. Doshi, T. Osadebe, L. Heldt, E.H. Chi, A. Beutel]{Flavien Prost, Ben Packer, Jilin Chen, Li Wei, Pierre Kremp, Nicholas Blumm, Susan Wang, Tulsee Doshi, Tonia Osadebe, Lukasz Heldt, Ed H. Chi, Alex Beutel}
\date{October 2022}

\usepackage{natbib}
\usepackage{graphicx}
\usepackage{caption}
\usepackage{subcaption}

\begin{document}

\begin{abstract}
  There has been a flurry of research in recent years on notions of fairness in ranking and recommender systems, particularly on how to evaluate if a recommender allocates exposure equally across groups of relevant items (also known as provider fairness). While this research has laid an important foundation, it gave rise to different approaches depending on whether relevant items are compared per-user/per-query or aggregated across users. Despite both being  established and intuitive, we discover that these two notions can lead to opposite conclusions, a form of Simpson's Paradox. We reconcile these notions and show that the tension is due to differences in distributions of users where items are relevant, and break down the important factors of the user's recommendations. Based on this new understanding, practitioners might be interested in either notions, but might face challenges with the per-user metric due to partial observability of the relevance and user satisfaction, typical in real-world recommenders. We describe a technique based on distribution matching to estimate it in such a scenario. We demonstrate on simulated and real-world recommender data the effectiveness and usefulness of such an approach.
\end{abstract}

\maketitle

\section{Introduction}
How can we understand and measure recommender performance and provider fairness \emph{across the broad diversity of users and requests} that they serve?
In recent years, the research community has recognized the importance and challenge of evaluating the fairness of recommender systems, with numerous proposals being put forth \cite{ashudeeppaper, Biega, fdd, alexbeutelrecommendation, sirui, DBLP:journals/corr/abs-2201-01180, 10.1145/3269206.3272027, DBLP:journals/corr/Burke17aa}. In this paper, we focus on the goal of provider-side fairness \cite{ashudeeppaper, Biega, fdd, alexbeutelrecommendation, sirui}: we assume that items or providers have different sensitive groups and want to measure whether the recommender provides equally good ranking/exposure for items of different groups.  Across many domains, from music recommendation \cite{10.1145/3269206.3272027} to advertising \cite{DBLP:journals/corr/abs-1806-06122} to job search \cite{linkedin, 10.1145/3442188.3445928}, such a property is of deep importance to the items and their producers.

We  study a recommender system which is responsible for ranking items for different users. Here, the term “user” refers to the recommendation event and is a similar concept as “requests” or “queries” in a recommender system. As users typically are exposed to the highest ranked items, the recommender plays an important role to allocate the traffic to items of different providers. In this context, there has recently been a proliferation of metrics proposed to capture the provider-side fairness of recommender systems \cite{alexbeutelrecommendation, gupta, ashudeeppaper, Biega, fdd, sirui}.  Many papers analyze the exposure allocated to items and propose that \textit{items from different groups should have equal exposure}, usually with a relation to their relevance. However we note one significant difference between previous metrics, which previous papers and reviews \cite{reviewFairness} have overlooked: some prior work compares the exposure of items per-user (or "per-query") \cite{ashudeeppaper, fdd_stochastic, Biega, sirui, linkedin} while others aggregate across users \cite{alexbeutelrecommendation, gupta}. Our first contribution is to highlight how this decision can lead to significant practical differences.

\vspace{-1mm}
\begin{figure}[htbp]
    \centering
    \includegraphics[width=0.52\textwidth]{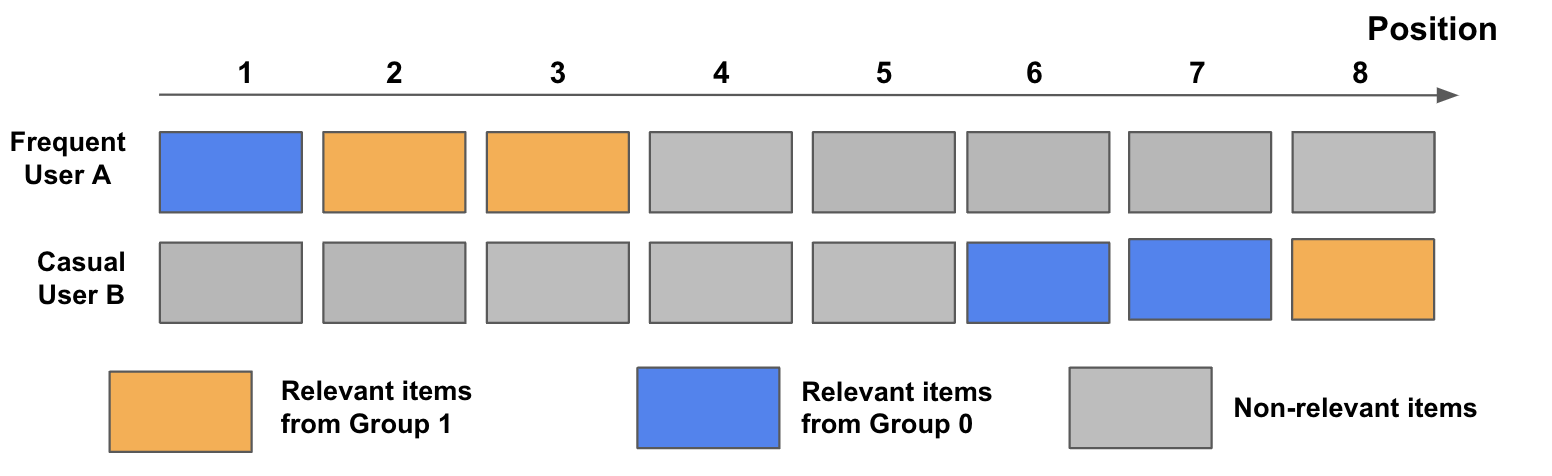}
    \caption{Simpson's paradox (two contradictory observations): For each user, Orange seems disadvantaged (ranked lower than Blue). However, aggregated across users, Orange is advantaged (Orange's mean position\protect\footnotemark is 4.3 compared to 4.7 for Blue).
    Explanation: Some of the user factors might affect exposure. In particular, the system ranks the relevant items (both Orange and Blue) in low position for casual users, for which Blue are frequent.}
    \label{fig:simpson}
\end{figure}

\footnotetext{For simplification of this example, we use the mean position as the exposure of a given group}

Let's make it more concrete through a motivating example in which we consider two commonly used  metrics from the fairness literature:
\begin{enumerate}
    \item Per-user: For each user, relevant items created by Group 0 providers should have an equal chance of exposure as relevant items from Group 1 providers \cite{ashudeeppaper, fdd_stochastic, Biega, sirui, linkedin}.
    \item Aggregate: Aggregated across all users, relevant items created by Group 0 providers should have an equal chance of exposure as relevant items from Group 1 providers \cite{alexbeutelrecommendation, gupta}.
\end{enumerate}

As both of these metrics seem intuitively similar and are supported by past literature, practitioners might choose between them arbitrarily. The first key contribution of this paper is the surprising finding that \textit{these metrics can actually point in opposite directions in practice} due to a Simpson’s paradox. Indeed, in Figure \ref{fig:simpson}, according to the per-user metric, Group 0 providers are advantaged (they are ranked higher than items from Group 1 on any user) while the aggregate metric indicates that Group 0 providers are disadvantaged (these items have lower exposure in aggregate).

How can we explain these seemingly-contradictory observations? We point out that the aggregate metric is influenced by \emph{which} users that each group is relevant. Indeed, when aggregating, we are comparing two data samples (the set of relevant items from each group) which follow different user distributions. Hence these differences of user distributions might influence the comparison, similar to a confounding variable in causal literature \cite{Simpson1951TheIO, 10.2307/2684093, doi:10.1126/science.187.4175.398}. For example, if one group of items is more often relevant to frequent users than casual users (where the recommender has lower \textit{performance}, i.e., putting relevant items in low positions), that group of items could seem to be advantaged in an aggregate evaluation while still facing a per-user gap, which is the scenario of Figure \ref{fig:simpson}. Alternatively, if one group of items is more relevant to users with broad interests for who many items are relevant (hence competing for exposure), the aggregate evaluation may suggest a bias against that group, when this is in fact due to this higher \textit{competition}. We derive theoretical equations to establish that, indeed, \textbf{the aggregate metric is influenced by the per-user metric but also by the user distributions (the competition and performance factors)}.

Given this new understanding, one may be interested in either per-user or aggregate evaluations depending on the application. While the aggregate metric can be computed under most scenarios, past literature on per-user evaluation \cite{ashudeeppaper, fdd_stochastic, Biega, sirui, linkedin} generally assumes that the relevance for every item is observed. However, in real-world recommenders, we often get user feedback on one, or a few items, out of a set of recommendations, as each user at each moment in time is unique. As a result, we rarely get to observe the relevance of multiple items on the same user, which does not allow to directly compute the per-user metric. To enable practitioners, we describe an approach to overcome this \textit{partial observability}: we show that we can estimate the per-user metric by matching the user distributions of the two samples being compared, and present a solution based on a sample matching strategy.

Finally, we study the conflict between these two metrics empirically and assess the quality of our technique to estimate the per-user metric under partial observability.  First, we use a simulated environment to demonstrate that distributional differences of the users can indeed lead to conflicting conclusions between the metrics. We show that our approach can accurately estimate the per-user metric under partial observability.  Second, we analyze a real-world recommender, similar to past work \cite{alexbeutelrecommendation,DBLP:journals/corr/abs-1708-05031, 10.1145/3219819.3220007}, and demonstrate that the tension between both metrics is observable in such a system, as each group of items is relevant to different users.  This demonstrates the importance of taking a more holistic view to study fairness in recommender systems.
 
To summarize, our contributions are:
\vspace{-1mm}
\begin{enumerate}
    \item \textbf{Simpson's Paradox in Ranking Fairness:} We highlight a conflict between two well-established fairness metrics in recommendations -- the per-user and aggregate metrics. We highlight that the aggregate metric is indeed influenced by differences in user distribution across the two groups, a form of Simpson's Paradox. We point out that the influence of the user distributions on the aggregate metric can be captured by two factors of the user's recommendations, competition and performance, and formalize this by characterizing mathematically the aggregate metric as a function of the per-user gap, the competition, and the performance (Section \ref{sec:Paradox}).
    \item \textbf{Method to estimate the per-user metric under partial observability:} In order to overcome the partial observability of relevance, we employ approaches from the causal literature to match the distribution of users across groups and estimate the per-user metric (Section \ref{sec:isolating_per_context}).
    \item \textbf{Experimental Evidence:} We demonstrate on multiple datasets that the tension between the two metrics might occur in both simulated and real-world environments, similar to the Simpson's paradox mentioned. Then we show the efficiency of our technique to recover the per-user metric (Sec. 6 and 7).
\end{enumerate}

\vspace{-3mm}
\section{Related Work}

\paragraph{Evaluation of Recommenders and Ranking}

To support the development of recommender systems, there has been much research to develop efficient ranking metrics. 
Offline evaluation methods \cite{DBLP:journals/corr/SchnabelSSCJ16, DBLP:journals/corr/AgarwalBSJ17, gilotte_2018, DBLP:journals/corr/abs-1003-5956, 10.1145/3289600.3291017} use logged data to measure how well a recommender exposes items to users. These approaches rely on a notion of \textit{relevant items} \cite{10.1145/3240323.3241622} and measure how well the model is able to predict the relevance of each item (e.g. Mean Squared Error \cite{DBLP:journals/corr/abs-1806-06122}) or how well it ranks relevant items (e.g. Normalized Discounted Cumulative Gain \cite{10.1145/582415.582418, 10.1145/345508.345545, 10.1145/582415.582418, ndcgcomparison}). The "relevance" can be based on human ratings but, for real-world personalized systems where such ratings are unavailable, we typically utilize implicit feedback and define relevance based on user engagement and satisfaction. This data collection from an already deployed recommender can lead to sampling biases \cite{DBLP:journals/corr/SwaminathanJ15, joachimsLondon}, and offline evaluations might therefore not reflect well the behavior of the actual system \cite{SimpsonsRecs}. To correct for this bias, prior work has utilized random experiments where items are randomly exposed to users, a setting also known as logged bandit feedback \cite{DBLP:journals/corr/SwaminathanJ15, joachimsLondon}. While these methods have laid important foundations for fairness metrics (e.g. the notion of relevance, the understanding on data biases), they focus on comparing recommender models, and do not extend directly to comparing groups of items or providers, which is necessary for group fairness evaluation.

\vspace{-1mm}
\paragraph{Fairness, Recommendation, and Ranking}

There has been a growing interest to ensure that recommenders provide fair recommendations for groups of items or providers from different demographic groups \cite{ashudeeppaper, Biega, fdd, alexbeutelrecommendation, sirui}.  Most of these papers follow the notion of \textit{Equality of Exposure}, where Exposure is typically defined as a function of the position of the items over the rankings \cite{10.1145/582415.582418}, such as NDCG \cite{Biega}, impression rate \cite{fdd_stochastic} or  pairwise order \cite{alexbeutelrecommendation, gupta}. Some papers suggest that the exposure should be allocated equally between groups \cite{StatisticalKuhlman}, mapping to Statistical Parity. Others follow the perspective of Equality of Opportunity \cite{HardtPS16}, accounting for relevance \cite{ashudeeppaper, Biega, fdd, alexbeutelrecommendation, sirui}.  For instance, \cite{ashudeeppaper, Biega} propose that the exposure of an item or group of items should be proportional to its average relevance. A variant to this proportionality constraint is to ensure that equally relevant items should have equal exposure \cite{alexbeutelrecommendation, sirui}. We will focus on this last option, however our observations transfer to the proportionality-based metrics.

Within this notion of Equality of Exposure, we group the past literature into two different categories: some only compare the exposure of relevant items per user \cite{ashudeeppaper, fdd_stochastic, linkedin, DBLP:journals/corr/abs-2108-05152}, while others compare the exposure when aggregated across users \cite{alexbeutelrecommendation, gupta}. As we show via the Simpson's paradox, these two different notions can lead to different conclusions: the difference is due to the fact that, when aggregating across users, the difference in user distributions between the two groups might have an influence on the gap. We will articulate this influence of the user by highlighting two specific factors related to the user: the performance, i.e., whether the system is able to rank relevant higher than non-relevant items for this user, and the competition, i.e., how many relevant items the user is interested in. We then prove that the aggregate metric is the combination of the per-user metric and these two factors.

This decomposition helps to structure prior literature. The first branch of per-user comparisons \cite{ashudeeppaper, fdd_stochastic, Biega, sirui, linkedin} measures a narrower concept and does not incorporate differences due to the other factors. For instance, it will not reflect a scenario where a group is relevant primarily to users where the system has poor performance. On the other hand, the aggregate comparison \cite{alexbeutelrecommendation, gupta} encapsulates these scenarios but without any disentanglement. Interestingly some prior work \cite{sirui, Biega} takes an intermediate path: they compare the exposure in aggregate but use an exposure normalized by the relevance of other items. The most common example is the Normalized Discounted cumulative gain (NDCG) \cite{sirui, Biega}, where the denominator accounts for how many items are also relevant for this user. This effectively controls for the impact of the competition factor on the metric. 

This framework can provide more insights to the practitioner when evaluating their recommender system. However the per-user metric, as well as the ones using normalized exposure, require knowing the relevance of all items for all users, and therefore are challenging to apply to real-world recommender. Indeed such systems often define relevance based on post-click signals, such as dwell-time \cite{dwell} or ratings \cite{5197422}, and users typically engage with no more than 1 item. This lack of observability is exacerbated by the use of random experiments \cite{joachimsLondon, alexbeutelrecommendation}, where an item is sampled from the corpus and shown. This process enables us to collect an unbiased sample of items with user feedback but limits as well the set of relevant items to maximum one per user (the boosted item if engaged). We define \textbf{partial observability}, as the situation where the relevance is known for only a few instances. \textit{We analyze the case of partial observability and suggest a technique, based on distribution matching, to estimate the per-user metric}. To the best of our knowledge, there has been no solution proposed before under this situation. 


\vspace{-2mm}

\paragraph{Distribution Matching}
We connect this challenge to a difference in underlying distributions where each group is sampled from different distributions of users. This relates to Simpson's paradox \cite{Simpson1951TheIO, 10.2307/2684093}, where a confounding variable can have an effect on the measured trend. A classic example comes from graduate admissions at UC Berkeley \cite{doi:10.1126/science.187.4175.398}: The overall rates showed that men were more likely to be admitted, but men were applying to departments with higher admission rates. When controlling for the department, the data showed a slight advantage towards women. A common solution to control for this effect is to equalize the distributions of both groups.
To that end, our paper builds on past reweighing techniques, based on Inverse Propensity Weighting or Scoring, commonly used to match two distributions in the causal literature \cite{Stuart_2010, 10.1093/biomet/70.1.41,causal2, causal3} and in the offline recommender evaluation work \cite{DBLP:journals/corr/SwaminathanJ15, joachimsLondon}.

\vspace{-2mm}
\section{Problem Setup and Background}

\subsection{Notations}

A recommender or ranking system is responsible for presenting the relevant items to the users. More formally, a recommender system ranks items $i$ from a catalog $I$ for any user $u$, where the \textit{user} encapsulates both user and contextual information \cite{contaware}, e.g. time of the day, user features, and is a similar concept to "query" in search systems. Note that, in some other applications, \textit{user} might be referred as event, request, query, context, or home page. One concrete example of such system is a comments section of a news site: When a user reads a certain news article, the recommender system would select the most relevant comments to present. In this example, the items are all the comments available and the user contains properties of the article (e.g. topic, language) as well as properties of the viewer (e.g. past history, country). The recommender considers all possible comments, ranks them and only the top ones are shown. Another example is a social media platform, where the user is shown the most relevant pieces of content on their page. Each item $i$ is a piece of content that the user might be interested in, the user is some combination of the user's profile and historic activity.

The relevance of an item $i$ for a given user $u$ is defined as $Rel(i, u)$, and is typically defined from user feedback \cite{10.1145/3240323.3241622}, such as clicks \cite{10.1145/2648584.2648589,41159}, ratings \cite{5197422} or dwell-time \cite{dwell}. It is important to note that the relevance is not a function of an item only but instead is defined for every item-cross-user $(i,u) \in I \times U$, representing a personalized and contextualized recommender system. However throughout this paper, we might use loosely the wording "relevant items" to designate all tuples $(i,u)$ where an item $i$ is relevant for the user $u$. For simplicity, we binarize the notion of relevance, $Rel(i,u) \in \{0, 1\}$, but it can be bucketed as well to reflect different levels of relevance, similar to \cite{10.1145/345508.345545, alexbeutelrecommendation, prost2019h}. Initially and until Section \ref{sec:isolating_per_context}, we assume that $Rel$ is known for all items and users $(i,u)$.

For every user, the recommender ranks all items and assigns them to a ranked-position $\mathit{Position}(i,u)$. Then items are \emph{exposed} to the user based on this ranking. While our work can be generalized to various models for $Exposure$, we decide here to model it as a function of the item's position: $\mathit{Exposure}(i, u) = f(\mathit{Position}(i,u))$. The function $f$ can for instance be a position-decayed function, as in NDCG or an impression prior \cite{10.1145/582415.582418}, to represent the probability that the item is seen by the user.

\vspace{-1mm}
\subsection{Two fairness metrics for recommenders}

We now assume that each item $i$ has a sensitive attribute $S(i) \in \{0, 1\}$  and our goal is to measure if items from different groups are treated equally by the recommender system. For example, in the news article comments setting, $S$ may be whether the comment's author belongs to a protected demographic group, and in the social media setting, $S$ may similarly correspond to the item producer's demographic group (e.g. gender) or, for a non-fairness use case, its topic.

In this paper, we follow the perspective of Equality of Opportunity from the classification literature \cite{HardtPS16} which has been adapted for ranking \cite{ashudeeppaper, Biega, fdd, alexbeutelrecommendation, sirui}. \cite{ashudeeppaper, Biega} propose that the exposure of a group should be proportional to its aggregated relevance, while \cite{alexbeutelrecommendation, gupta} uses conditioning to capture that items, when relevant, should have equal exposure. We will use this latter notion for our definitions, but we expect the observations transfer to the other definitions. We will now describe  definitions following both frameworks for use in the rest of the paper.

\begin{definition}
\label{def:peruser}
\textbf{For a user}, the per-user gap is defined as the difference between groups in expected exposure over items that are relevant for this given user.

\vspace{-2mm}

\begin{align*}
\begin{split}
PerUserGap (u) &= \mathop{{}\mathbb{E}}_{i\sim I} [Exposure(i, u) | Rel(i, u) = 1, S(i) = 1] \\
&- \mathop{{}\mathbb{E}}_{i\sim I} [Exposure(i, u) | Rel(i, u) = 1, S(i) = 0]
\end{split}
\end{align*}
\end{definition}

As this metric is defined for each user, we can capture the full system gap by averaging over users: $PerUserGap = \mathop{{}\mathbb{E}}_u PerUserGap(u)$.

\begin{definition}
\label{def:aggregate}
The aggregated gap is defined as the difference between groups in average exposure over over items that are relevant for any user.

\vspace{-2mm}

\begin{align*}
\begin{split}
AggregateGap &= \mathop{{}\mathbb{E}}_{(i,u)\sim I \times U} [Exposure(i, u) | Rel(i, u) = 1, S(i) = 1] \\
&- \mathop{{}\mathbb{E}}_{(i,u) \sim I \times U} [Exposure(i, u) | Rel(i, u) = 1, S(i) = 0] 
\end{split}
\end{align*}
\end{definition}

Equivalently, we could remove the conditioning and rewrite the metric as the difference of exposure over two uniform samples from $D^0$ and $D^1$, with $ D^s = \{ (i,u) | Rel(i,u) = 1, S(i) = s \}$. This metric targets that, aggregated over the instances where an item is relevant to a user, the exposure should not depend on the group.

Each definition is supported by prior work and seems intuitive when taken separately. However we will see that they represent different notions and articulate those.

\vspace{-2mm}
\section{Reconciling the metrics}
\label{sec:Paradox}

\subsection{Simpson's paradox as a motivating example}
\label{sec:simpsons}

Let's first go through the example of Fig. \ref{fig:simpson} to articulate the potential tension between these two  metrics. Let's imagine that this is a recommender in a social media setting, which has to rank items for two different users: the first one is a frequent user who uses the social media regularly, while the other one is a casual user. We display items in Orange/Blue (depending on their group) if they are relevant for that user, and in Grey if they are not. To simplify the computations for this example, let $Exposure(i, u)$ simply be $Position(i,u)$, such that the exposure gap measures the difference in average ranked position for the relevant items from each group. 

The Simpson's paradox comes from making two seemingly contradictory observations. (1) For any given user, Orange is ranked lower than Blue, i.e., Orange seems disadvantaged according to the per-user metric in Def. \ref{def:peruser}. (2) When aggregating across users, the Orange group has better exposure; the average position of Orange items is 4.3, compared to 4.7 for Blue, i.e., Orange seems advantaged according to the aggregate metric in Def. \ref{def:aggregate}.

Intuitively in this example, this paradox is due to the effect of the users on the exposure of items -- in our example, the recommender does a better job at ranking the relevant items (Blue \& Orange) above non-relevant ones (Grey) for active users than for casual ones; this could be due to having less historical data for the latter. Combined with the fact that the Blue group has more relevant content for casual users, this offsets the fact that Blue is always ranked above Orange within the same user.

\subsection{Qualitative analysis}

In this section, we show that the aggregate metric is broader as it is influenced by the per-user gap but also by two additional factors per-user: the performance and the competition. This influence is due to the fact that this aggregate gap compares groups over two different user distributions.

\paragraph{The aggregate metric compares two samples with different user distributions.} In the aggregate metric, we are comparing two samples of relevant items, $D^0$ and $D^1$ with $ D^s = \{ (i,u) | Rel(i,u) = 1, S(i) = s \}$, which have different marginal distributions over users. Indeed, as each sample is defined as a uniform sample from $D^s$, the marginal distribution over users will be proportional to $Count(i | Rel(i, u)=1, S(i) =s)$.

Due to these different marginal distributions, the aggregated exposure of each item/provider group will be differently impacted by the users where it is relevant. This user influence on the aggregate metric acts like a "confound" in causal literature \cite{hernan2010causal, doi:10.1126/science.187.4175.398}. This is what we saw in the example before, where the Blue group is more frequent in casual users, where relevant items have on average lower exposure.

\vspace{-1mm}
\paragraph{Decomposing the user influence into two factors.}

Let's break down the exposure of relevant items for a given user into an intuitive set of factors; The first one is the per-user gap as defined above (we will define them rigorously in the next paragraph), while the other two are additional per-user factors.
\vspace{-1mm}
\begin{enumerate}
    \item \textbf{Per-user:} As in Def. \ref{def:peruser}, this is the difference in exposure between relevant items of different groups on a given user.
    \item \textbf{Performance:} The capacity of the recommender to rank relevant items higher than non-relevant ones on a given user. If the system has higher performance for one user than for another, items relevant to the first user will have higher exposure.
    \item \textbf{Competition:} The number of items that are relevant to a given user. As the exposure is limited (only a few items are shown to the user), the average exposure of relevant items is reduced when many items are relevant to a given user.
\end{enumerate}

This breakdown allows us to identify three separate scenarios under which a Group (let's say Orange) might have lower exposure according to the aggregate metric. Fig \ref{fig:percontext1} shows a situation where there is a per-user gap, which translates into an aggregate gap. However, Fig \ref{fig:performance1} and \ref{fig:competition1} point to different phenomena: while both do not display any per-user gap, the aggregate metric will indicate that Orange is disadvantaged. Indeed, in Fig. \ref{fig:performance1}, the recommender system is not as good at allocating exposure to relevant items for users 1 and 2, and it turns out that the Orange group appears primarily on these users. As a result, Orange still receives lower exposure because it appears primarily on users with lower performance. Competition suggests another scenario displayed in Figure \ref{fig:competition1}: there is more competition on users 1 and 2 (two relevant items, versus just one) and it turns out that the Orange group is primarily relevant to these users. As a result, Orange gets lower exposure because it appears primarily on users with higher competition. From this, we see that lower exposure in aggregate can be due to any of these three factors.

\begin{figure}
    \begin{subfigure}[t]{0.46\textwidth}
      \centering
        \includegraphics[width=\textwidth]{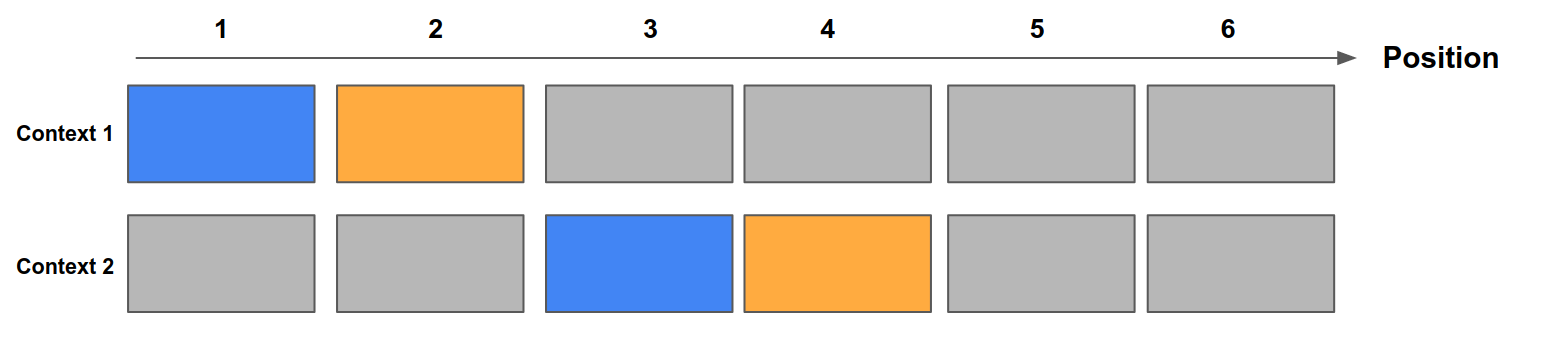}
        \caption{There is a per-user gap, which causes a gap in aggregate.}
        \label{fig:percontext1}
    \end{subfigure}
    \begin{subfigure}[t]{0.46\textwidth}
      \centering
        \includegraphics[width=\textwidth]{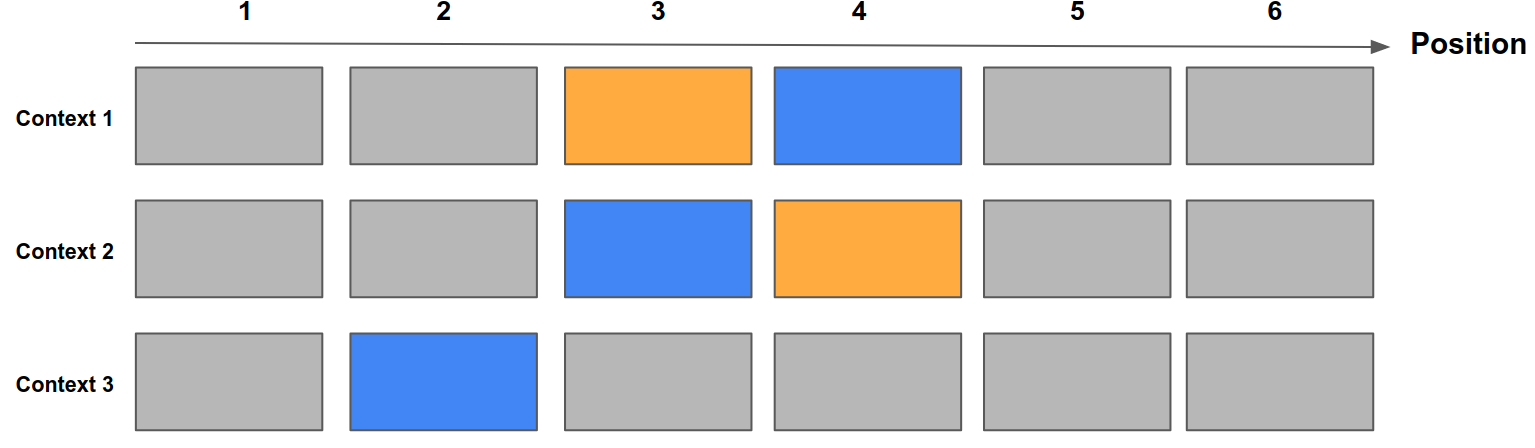}
        \caption{Without per-user gap, Orange still gets lower exposure in aggregate as it appears more on users with lower performance.}
        \label{fig:performance1}
    \end{subfigure}
    \begin{subfigure}[t]{0.46\textwidth}
      \centering
        \includegraphics[width=\textwidth]{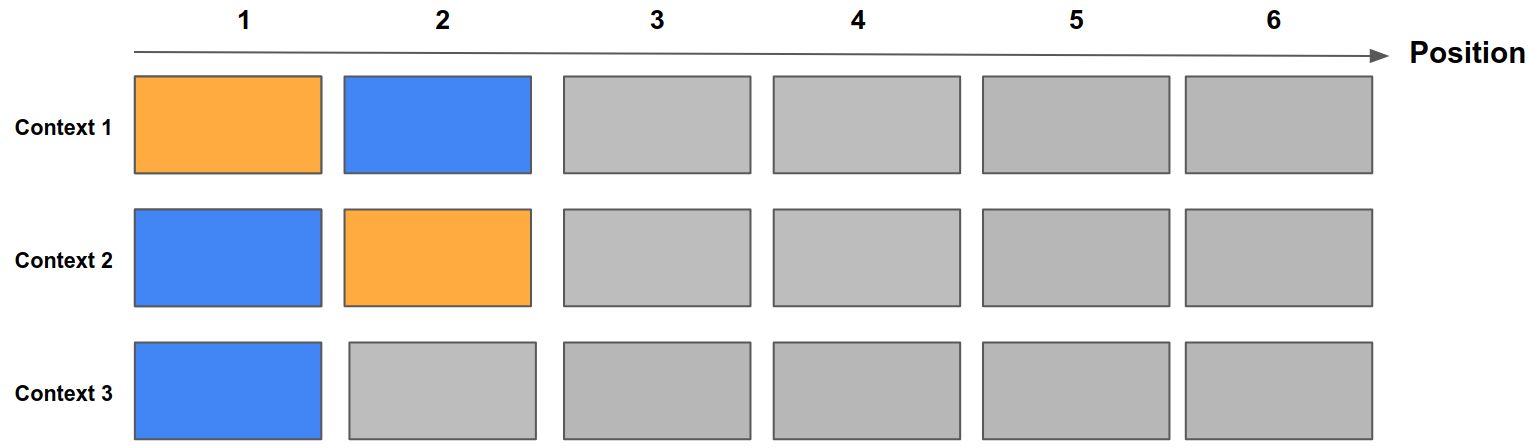}
        \caption{Without per-user gap, Orange still gets lower exposure in aggregate as it appears more on users with higher competition.}
        \label{fig:competition1}
    \end{subfigure}
    \caption{Three separate scenarios with an $AggregateGap$ against Orange. Only the first one is associated with a Per-user gap.}
    \vspace{-5mm}
\label{fig:test}
\end{figure}

\subsection{Theoretical analysis}

We now present theoretical evidence to prove that the qualitative decomposition of the aggregate metric into these three factors is well-founded. Let's start with some preliminary definitions:

\begin{definition}
We define the \textbf{performance} for a given user as the difference of exposure between relevant and irrelevant items.

\vspace{-1mm}
\begin{align*}
\begin{split}
\mathit{Performance}(u) &= \mathop{{}\mathbb{E}}_i [Exposure(i, u) | Rel(i, u) = 1] \\
&- \mathop{{}\mathbb{E}}_i [Exposure(i, u) | Rel(i, u) = 0]
\end{split}
\end{align*}
\end{definition}

\begin{definition}
We define the \textbf{competition} as the number of relevant items from group $s$ for a given user $u$.
\begin{align*}
\begin{split}
Comp_s (u) = Count(i| Rel(i,u) = 1, S(i) = s)
\end{split}
\end{align*}
\end{definition}

\paragraph{Decomposition of the aggregate metric}

We now prove that the competition and performance are sufficient to capture the additional influence of the user distribution on the aggregate metric. We rely on the following theorem, as the key theoretical foundation, which decomposes the aggregate metric as a function of the 3 terms.

\begin{theorem}\label{maintheorem} We can rewrite the $\mathit{AggregateGap}$ as follows:

\vspace{-4mm}

\begin{align*}
\begin{split}
&\mathit{AggregateGap} = \sum_{u} \Biggl( Z* 
 \frac{Comp_0(u)* Comp_1(u)}{Comp_1(u) + Comp_0(u)} * PerUserGap(u) \\
& + (\frac{Comp_1(u)}{\sum_{u} Comp_1(u)} - \frac{Comp_0(u)}{\sum_{u} Comp_0(u)}) * NR(u) * \mathit{Performance}(u)
\Biggr) \\
\end{split}
\end{align*}

\end{theorem}
\noindent where $NR(u) = (1 - \tfrac{Comp_0 (u) + Comp_1 (u)}{|I|})$ is the proportion of items that are not relevant for user $u$, and $Z = \frac{1}{\sum_{u} Comp_1(u)} + \frac{1}{\sum_{u} Comp_0(u)} $ is a constant. The full proof is included in the appendix.

This theorem highlights how the $PerUserGap$, combined with the performance and competition, form the $AggregateGap$. We can share some high-level intuition about the different terms.

\begin{itemize}
    \item The first term represents the contribution of the $PerUserGap$. Users with a balanced amount of relevant Orange and Blue items (i.e., $ Comp_0(u) \cdot Comp_1(u)$ is high) will be weighted heavily, as they are the ones where relevant items from both groups compete the most for exposure.
    \item The second term represents the contribution of the performance. The first factor $(\frac{Comp_1(u)}{\sum_{u} Comp_1(u)} - \frac{Comp_0(u)}{\sum_{u} Comp_0(u)})$ indicates that high performance contributes more to a positive gap when items from group 1 are more frequently relevant to the user. The second factor $NR(u)$ shows that performance matters more when there are more non-relevant items.
\end{itemize}

Interestingly, \textit{our factor decomposition gives us a new framework to grasp conceptual differences within past literature}. In prior work using the aggregate version \cite{alexbeutelrecommendation, gupta}, the per-user gap, the performance and the competition will all affect the metric without any disentanglement, as displayed in Fig. \ref{fig:test}. On the other hand, the past work on per-user (or per-query) metric \cite{ashudeeppaper, fdd_stochastic, linkedin} ignores the potential differences in the performance and competition and their effect. Finally, some of the literature has incorporated knowledge from information retrieval and used an exposure normalized by the relevance of other items. The most common example is the Normalized Discounted cumulative gain (NDCG) \cite{sirui, Biega}, where the denominator accounts for how many items are also relevant for this user. This effectively controls for the impact of the competition factor on the gap, but as we will discuss in the next section is hard to use in practice. We hope that this decomposition can guide the practitioners to understand their evaluation choices more precisely.

\vspace{-2mm}
\section{How to estimate the per-user gap under Partial Observability}
\label{sec:isolating_per_context}

\subsection{Challenge of Partial Observability}

In the previous section, we provided some decomposition to reconcile the two metrics: we outlined that the $AggregateGap$ was combined multiple effects -- the $PerUserGap$ itself, but also the performance and the competition. While we leave the decision to each application and practitioner, we now want to provide them with some technical tools for the metric of their choice in any situations.

The challenge lies in the data collection for real-world recommenders for two main reasons. First, the relevance is typically defined from user feedback including post-click feedback (e.g. are users are satisfied with the item when recommended) \cite{dwell, 5197422} and we typically observe the post-click behavior typically on only 1 item per user. Secondly, to remove any data bias, we typically measure this user feedback through random experiments where an item is sampled from the corpus and shown . This process, highlighted by prior work \cite{joachimsLondon, alexbeutelrecommendation}, enables us to collect a sample of items with user feedback but limits as well the set of relevant items to maximum one per user (the boosted item if engaged). Since this process leads to a uniform sample over $I \times U$, we can compute directly $AggregateGap$. However per-user metrics require knowing the relevance of all items (or at least multiple items) for all users and can therefore not be directly applied to real-world recommenders.

We define \textbf{partial observability}, as the situation where the relevance is known for only a few instances ($Rel(i, u)$ is sparse). \textit{The rest of this section focuses on how we can estimate the per-user metric under such situations, which has not been studied yet}.

\subsection{Sampling from a common user distribution.}

We now propose a framework to isolate the per-user gap under partial observability. As we have seen, the difference of marginal distributions is the reason why some factors of the users, namely the performance and competition, are affecting the gap. \textbf{A natural idea therefore is to match the user distribution of both groups}, which is supported by the following theorem.

\begin{theorem}
Let's assume that our samples of relevant items follow the same marginal distribution over users ( $D^s_{common}$ parameterized by $w_{common}$\footnote{we will discuss later reasonable choices for $w_{common}$}). Then we can estimate the $PerUserGap$ by comparing the average exposure over these two data samples.

\end{theorem}

\begin{proof}

We are taking the difference of two expectations over the same distribution of users and can use the linearity of the expectation to re-arrange the term.

\begin{align*}
\begin{split}
&\mathop{{}\mathbb{E}}_{(i,u) \sim D^1_{common}} [Exposure(i, u)| Rel(i,u) = 1, S(i) = 1]\\ &-\mathop{{}\mathbb{E}}_{(i,u) \sim D^0_{common}} [Exposure(i, u)| Rel(i,u) = 1, S(i) = 0] \\
&= \mathop{{}\mathbb{E}}_{u \sim w_{common}} \mathop{{}\mathbb{E}}_i [Exposure(i, u)| Rel(i,u) = 1, S(i) = 1] \\
&- \mathop{{}\mathbb{E}}_{u \sim w_{common}} \mathop{{}\mathbb{E}}_i [Exposure(i, u)| Rel(i,u) = 1, S(i) = 0] \\
&= \mathop{{}\mathbb{E}}_{u \sim w_{common}} PerUserGap(u)
\end{split}
\end{align*}

\end{proof}

This theorem is intuitive. If our data samples have the same distribution over users, then the user factors will not influence the metric. Matching the distribution of groups over some variables is a common technique in the causal literature to control for the effect of confounding variables \cite{Stuart_2010, 10.1093/biomet/70.1.41,causal2, causal3}.

\vspace{-2mm}
\subsection{Reweighting technique to match distributions}

The previous theorem gives us a recipe to estimate the per-user metric by ensuring that our two data samples have the same distribution. As modifying the data collection might be challenging, we instead suggest a reweighting approach, similar to Inverse Propensity Scoring (IPS), to match the distributions offline.

We consider some features of a user such as user activity, time or country. The choice of these variables is application-dependent and needs to be interpreted in the particular practical setting, similar to the discussion of selecting resolving variables in \cite{resolving, explainablediscrimination}: Some might want to use a fine-grained representation of the user to control for any potential influence on the gap, while some others might select only a subset of the user characteristics.

We now assume that we have a list of user variables $V$ that we want to control for. As mentioned earlier, the two groups have different distributions over $V$, which therefore influences the measure of the gap. The distribution over the features $V$ for one group $s \in \{0, 1\}$ can be written as follows:

\vspace{-1mm}
\begin{align*}
\begin{split}
\forall\ v\ \in\ V, w_s(v) & = \mathbb{P}_{i, u}(V(u) = v\ |\ Rel(i,u) = 1, S(i)= s) \\
&\propto \mathbb{P}_{i, u}(V(u) = v,\ Rel(i,u) =1,\ S(i)= s)
\end{split}
\end{align*}

\textbf{Our method consists in matching the two distributions,  $w_0$ and $w_{1}$, by reweighting the data based on IPS}. This technique has been used as to remove data biases in offline recommender evaluation \cite{DBLP:journals/corr/SwaminathanJ15, joachimsLondon}, or to control for the effect of confounding variables in the causal literature \cite{10.1093/biomet/70.1.41,causal2, causal3}.

To match the two distributions, we start by selecting a target distribution over the variables V. While various target distributions could be used\footnote{For instance, in addition to the uniform distribution presented in this paper, one might choose to use the empirical distribution of the sensitive group.}, we take here a simple approach and assume a uniform target distribution. In other words, we want to match each group such that the distribution over V is uniform: $  \forall\ v\ \in\ V, w_{common}(v) = 1$. We can then use IPS to match the data to this target distribution by re-weighting the collected data of each group $s$ by $w_{common}(v) / w_s(v)$.

\vspace{3mm}

\textbf{Estimation of the weight.}
If V is simple (e.g. only a few variables that can be bucketized easily), $w_s(v)$ should be computable by empirical counts such as:
\vspace{-2mm}
\begin{align*}
\begin{split}
&\forall\ v\ \in\ V, w_s(v) = \frac{|(i,u) | Rel(i,u) =1, S(i)= s, V(u) = v)|}{|(i,u) | Rel(i,u) =1, S(i)= s|}
\end{split}
\end{align*}

If v is more complex (e.g. we want to correct for a large number of user characteristics, or v can not be bucketized), this estimation might not be possible due to data sparsity. In this case, we can rewrite the source distribution as follows:
\vspace{-0.5mm}

\begin{align}
\label{ipseq}
\begin{split}
w_s(v) &\propto \mathbb{P}_{i, u}(V(u) = v, Rel(i,u) =1, S(i)= s) \\
  &\propto \mathbb{P}_{i,u}(V(u) = v, S= s) * \mathbb{P}_{i,u}(Rel(i,u) =1|V(u) = v, S(i)= s)
\end{split}
\end{align}

As the first term not depend on relevant items, it does not suffer from the partial observability issue and can be computed exactly from the data. The second term is the most challenging term and we can use a model to estimate it. To that purpose, we can use the random data, take as input all confounding variables V and predict the probability of an item being relevant\footnote{If V is a fine grained representation of the user, the prediction task might be challenging and the model might be imperfect (it is close to re-engineering the recommender system). However the main role of this model is to correct for potential confounding variables, rather than matching perfectly the target distribution.}.

\section{Simulation}

In this section, we use a simulated recommender system to demonstrate that the Simpson's paradox is observable under reasonable conditions. In particular, we show that $PerUserGap$ and $AggregateGap$ can point in opposite directions, and more generally that varying each of the three factors individually can cause large variations in $AggregateGap$. Furthermore, we show that our technique enables us to estimate the $PerUserGap$ under partial observability.

\subsection{Simulation setup}

We first explain at a high level the setup of the recommender. This design enables us to vary the per-user gap, the performance, and the competition as needed, and then observe the resulting aggregated gap in the system. To simulate these three factors, we have parameters for the density, the noise, and the model bias.

\vspace{2mm}
\textbf{High-level description:}
We simulate a recommender system composed of N users and M items. We categorize these items into 2 different groups.
We generate the data from a low-rank representation, similar to past work \cite{alexbeutelLowRank}: Each user (resp. item) can be represented with a normalized r-dimensional vector $U_i$ (resp. $V_j$). The relevance of an item for a user is then equal to the product of both representations: $U_i * V_j$. To map it to our problem, we binarize the relevance by setting a threshold: $Rel = 1$ if $U_i * V_j > T$ else 0. We arbitrarily set a threshold of T=0.3, which empirically maps to about 20\% of the items being relevant. Finally, the recommender scores are a noisy and biased estimate of the true scores: $Score[i, j] = U_i * V_j + noise[i, j] + bias[j]$.

The recommender system ranks the items based on $Score[i, j]$. We represent the exposure allocated to an item as the discounted cumulative gain of a position: $exposure(position)=\tfrac{1}{log_2(1+ position)}$. Finally, we reproduce the setup of the random experiment: for each user, one item is selected at random and we observe the true relevance. The aggregate metric is computed as the difference of exposure between the relevant items of each groups.

\vspace{2mm}
\textbf{Choice 1: How to generate the item representations [affecting the competition factor]} For users, we simply sample uniformly in $[-1, 1]^r$. However, for items, it is key to ensure that the two groups have different distributions of item representations. There are many ways to represent two different distributions over the latent space and we decide here to sample items along two different directions in the latent space. More precisely, for each group s in \{0, 1\}, we generate items from Gaussian($\mu_s$, $\frac{1}{d_s}$). We choose arbitrarily $\mu_0$ and $\mu_1$ to be orthogonal ([1, 0, .., 0] and [1, 1, .., 0]) but findings would adapt to other choices. More importantly, $d_s$ controls the density of each group and affects the competition factor. Indeed if a group is dense ($d_s$ large), all items are similar to each other and a user interested in a given item will also be interested in many other ones (high competition). We set $d_0$=1.5 (giving some overlap between the group distributions) and vary only $d_1$.

\vspace{2mm}
\textbf{Choice 2: How to generate the noise [affecting the performance factor]} We also want our simulation to represent a recommender that has different levels of performance for the two groups.
To that end, we use a default level of noise, $\sigma$, but make it vary alongside the direction $\mu_1$, such that the typical items from group 1 (aligned with $\mu_1$) will have much higher level of noise. This difference of noise is controlled by a parameter $\sigma_1$.
More precisely: $Noise[i, j ] = Gaussian( 0,\ \sigma * [1 + f(U_i * \mu_1)])$. We set $\sigma=0.05$ (arbitrarily), and F is a linear function so that f(-1)=0 mapping to items that are the opposite of typical items from group 1 and f(1)=$\sigma_1$ mapping to typical items from group 1. As a result, the level of noise will vary between $\sigma$ and $\sigma * (1+ \sigma_1)$ for typical items of group 1.

\vspace{2mm}
\textbf{Choice 3: How to represent the model bias [affecting the per-user metric]}
We represent the model bias by adding a constant term to the predictions when items are from group 1.

\vspace{-2mm}

\subsection{Variations of PerUserGap, Performance and Competition changes the AggregateGap}

First, we conduct some analysis to verify that our simulation enables us to affect the performance, the competition and the $PerUserGap$ of the recommender system. We vary the parameters of the system (the model bias, the level of noise and the density) and see their effect on the three factors (per-user gap, performance, competition). While these factors are defined for every user $u \in U$, we extend their definition to an aggregate level via the following methodology: (1) the per-user gap factor is reported as the average PerUserGap(u) over all users where at least one item from each group is relevant; (2) the competition and performance factors are computed for a hypothetical user that would be aligned with $\mu_1$, to reflect the characteristics that will dominate for the items from group 1.

Results are displayed in Fig \ref{fig:simulation1}. We confirm that the 3 system parameters are effective in varying the corresponding 3 factors. Indeed, the first graph shows that we can affect the per-user gap by changing the model bias and holding the other 2 parameters constant. The two other graphs make similar observations for the performance and competition.

\vspace{-1mm}
\begin{figure}[htbp]
    \centering
    \includegraphics[width=0.52\textwidth]{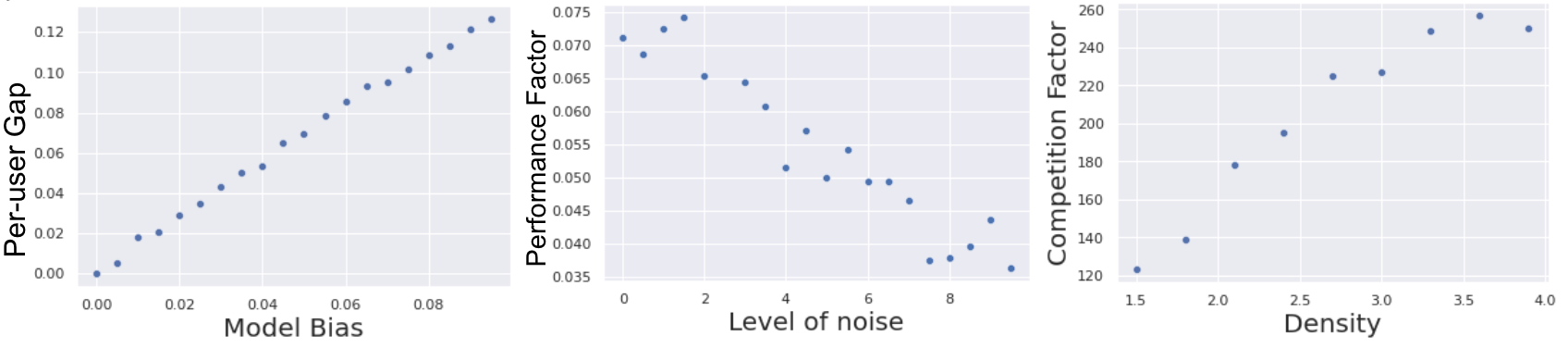}
    \caption{Each parameter (model bias, noise, density) varies the targeted factors (per-user gap, performance, competition).}
    \label{fig:simulation1}
    \vspace{-2mm}
\end{figure}

Then we emphasize in Fig. \ref{fig:simulation2} how the $AggregateGap$ is affected by these three factors. For instance, we see (left graph) that increasing the model bias (i.e., increasing the $PerUserGap$) results in an increase of the aggregate metric. This experiment provides empirical evidence that the performance, competition, and per-user gap affect the aggregate metric. In particular, the $AggregateGap$ can vary greatly even when the $PerUserGap$ is zero or fixed.

\vspace{-2mm}
\begin{figure}[htbp]
    \centering
    \vspace{-1mm}
    \includegraphics[width=0.52\textwidth]{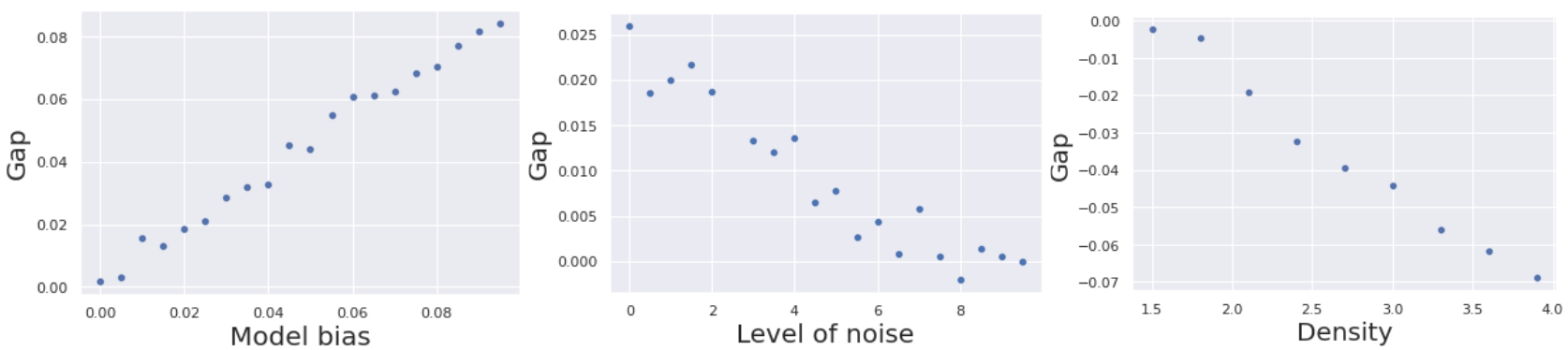}
    \caption{Each parameter will influence $AggregateGap$.}
    \label{fig:simulation2}
\end{figure}

\vspace{-2mm}

\subsection{Estimating the per-user metric}

We now want to confirm that the method described in Section \ref{sec:isolating_per_context} can estimate the per-user metric well. To do so, we reweigh the data based on eq. \ref{ipseq}: the first term is constant ($U_i$ is sampled uniformly) and we estimate the second term by a model that takes as input the user embedding $U_i$ and predicts the average relevance of items of each group for that user.

We first conduct an experiment to verify that the method is effective even when the competition or performance affect the aggregate gap significantly. To that end, we arbitrarily set a constant model bias of 0.03 and then vary the level of noise and density. Results are shown in Figure \ref{fig:Sim2}. On the left graph, we observe that as the density increases, the aggregate gap increases as well. We then plot the per-user gap (both true and estimated) and observe that our estimate is close to the true value. Similar observations can be made in the right graph when we vary the level of density. This experiment shows that our estimate of the per-user gap is reliable even when the two other factors have a large influence on the gap.

\begin{figure}[htbp]
    \centering
    \includegraphics[width=0.5\textwidth]{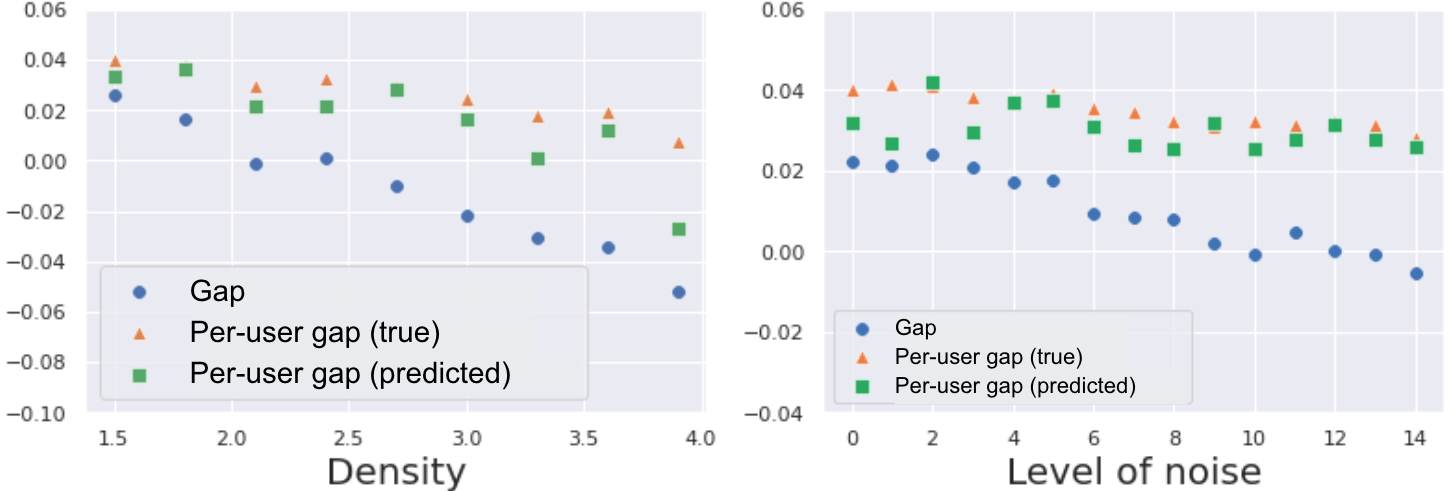}
    \caption{True vs Predicted PerUserGap depending on density/noise.}
    \label{fig:Sim2}
\vspace{-3mm}
\end{figure}

We further assess the reliability of our estimate by varying performance and competition at the same time. We again set a constant model bias of 0.03 and vary the other two parameters over a grid of values (model noise from 0 to 10 and density from 1.5 to 4.). Over 160 simulations, we see that the $AggregateGap$ is significantly different from the true $PerUserGap$ (4.0pp difference on average), which reflects again that the two other factors have a strong impact on the gap. Comparatively, our estimated $PerUserGap$ is a relatively robust estimate of the true value (0.9pp difference on average). This show that our method is able to estimate $PerUserGap$ effectively.

\section{Real-world recommender systems}

We now study a real-world recommender system: for every user (here the "user" includes contextual features such as time), the system predicts the relevance of all items in the corpus based on past interactions (click and post-clicks signals) and returns a ranked list of items. This is  a similar setup to \cite{alexbeutelrecommendation, DBLP:journals/corr/abs-1708-05031, 10.1145/3219819.3220007}. Only the top items are displayed to the user. We have available the group information of a subset of providers and can split them into a binary groups $S=\{0, 1\}$. Our goal is to evaluate offline whether the system provides equal treatment based on the sensitive group $S$. We want to estimate the aggregate metric and define the terms as follows.

\begin{itemize}
    \item \textit{Relevance}. Per product definition, an item is relevant for a user if, had the item been recommended, the user would have clicked on and been satisfied with the item. In order to access this in practice, we rely on random experiments where, for a subset of requests, an item is picked randomly and boosted high in the ranking. We then filter to boosted items with high user satisfaction to build our data sample.
    \item \textit{Exposure}. We estimate the exposure of items that were boosted and engaged based on the initial position before boosting. We use the impression-prior aggregated over historical data to estimate the chance of impression at this initial position.
\end{itemize}

We report in Fig. \ref{fig:graphReal} the aggregate exposure for items that the user engages with when recommended to them, depending on the sensitive group. We observe that relevant items from $S=1$ receive lower exposure on average, mapping to an $AggregateGap$ against this group. Note that due to the sensitive nature, we cannot report absolute exposure numbers. However we keep the y axis constant, so that the graphs can indicate the direction and scale of the difference in exposure between groups.

Then we want to understand if the  $AggregateGap$ is affected by differences in users for the groups, as illustrated by the Simpson's paradox of this paper. As the data collection relies on random experiment, we do not observe the relevance for every items (conditions of \textit{partial observability}) and therefore followed the approach described in Section \ref{sec:isolating_per_context}. Based on product experience, we identified some key characteristics of a user (e.g. some features related to the user past behavior or the related query) and used the technique presented to control for the influence of these variables.

\begin{figure}[htbp]
    \centering
    \includegraphics[width=0.3\textwidth]{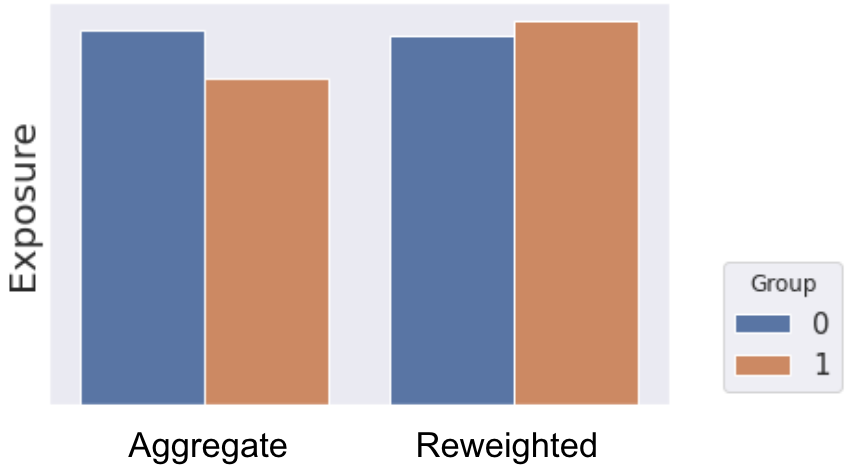}
    \caption{Exposure of 2 groups in a real-word recommender. Left: $AggregateGap$. Right: Estimation of $PerUserGap$ by reweighting.}
    \label{fig:graphReal}
\end{figure}

We report in Figure \ref{fig:graphReal} (right part) the exposure of each group after we control for these variables. The difference between both groups, which maps to the estimated per-user metric, is inverted compared to the gap from the $AggregateGap$: when comparing over users with the same characteristics, defined by the set of variables, relevant items from $S=1$ have actually higher exposure.

These results represent a real-world illustration of the Simpson's paradox, which our methodology has been able to exhibit: relevant items from  $S=1$ have lower $AggregateGap$, but have have higher exposure when compared on requests with similar characteristics which is the estimated $PerUserGap$.

\section{Conclusion}

Past literature has established two different fairness notions, comparing the exposure of relevant items either on a given user, or aggregated across users. While both notions seem intuitive, we illustrate the tension between them through a Simpson's paradox, where they lead to opposite conclusions. We provide a framework to articulate these differences: as the aggregate metric compares two groups with different user distributions, it will be influenced by the users. We show that the aggregate metric can be expressed as a function of the per-user metric and two user factors (competition and performance). This decomposition helps us to reconcile the conceptual differences within past literature. We then propose a method to estimate the per-user metric in the common situation of only partially observable relevance. This technique matches the distribution of users over the groups, which controls for the effect of performance and competition. We use simulations to show that this method is effective to estimate the per-user gap in practice.

This research shows that studying fairness in recommender system is complex and multi-faceted. While this work presents some foundations, it opens the door for more targeted remediation. It is indeed possible that one might take different actions depending on whether a gap is due to the performance, competition, or per-user factors.

\textbf{Acknowledgements:} The authors would like to thank Fernando Diaz and Konstantina Christakopoulou for their valuable feedback on this paper.

\bibliographystyle{plain}
\bibliography{references}

\newpage
\clearpage
\appendix

\newpage
\onecolumn
\section*{Appendix: Proof of Theorem \ref{maintheorem}}
\textbf{Step 1: Exposure of relevant (vs irrelevant) items on a given user.} 

The performance controls how the total exposure available (noted as ExposureAvailable) is allocated to relevant/irrelevant items, as we see in the following equations (using the definition of Performance).

For a given user $u$,
\begin{align*}
\begin{split}
ExposureAvailable =& \ (Comp_0(u) + Comp_1(u)) * \mathop{{}\mathbb{E}}_i [Exposure(i, u) | Rel(i,u) = 1] \\
&+ (|I| - Comp_0(u) - Comp_1(u)) *
\mathop{{}\mathbb{E}}_i [Exposure(i, u) | Rel(i,u) = 0] \\
ExposureAvailable =& \ |I| * \mathop{{}\mathbb{E}}_i [Exposure(i, u) | Rel(i,u) = 1] \\
&- (|I| - Comp_0(u) - Comp_1(u)) * Performance(u)
\end{split}
\end{align*}

Finally, we have:
\begin{align*}
\begin{split}
\mathop{{}\mathbb{E}}_i [Exposure(i, u) | Rel(i,u) = 1] =& \ \frac{ExposureAvailable}{|I|} \\
&+ (1 - \frac{Comp_0(u) + Comp_1(u))}{|I|}) * Performance(u)) 
\end{split}
\end{align*}

\vspace{5mm}
\textbf{Step 2: Exposure of relevant items from a given group on a given user.}

The expected exposure of relevant content is a weighted sum over the exposure of the two groups. The following equations follow (using the definition of Per-user Gap).

For a given user $u$,
\begin{align*}
\begin{split}
\mathop{{}\mathbb{E}}_i [Exposure(i, u) | Rel(i,u) = 1] =& \
  \frac{Comp_1(u)}{Comp_1(u) + Comp_0(u)} * \mathop{{}\mathbb{E}}_i [Exp(i,u) | Rel(i,u) = 1, S = 1] \\
& + \frac{Comp_0(u)}{Comp_1(u) + Comp_0(u)} * \mathop{{}\mathbb{E}}_i [Exp(i,u) | Rel(i,u) = 1, S = 0] \\ 
 \mathop{{}\mathbb{E}}_i [Exposure(i, u) | Rel(i,u) = 1] =& \ \mathop{{}\mathbb{E}}_i [Exp(i,u) | Rel(i,u) = 1, S = 1]  - \frac{Comp_0(u)}{Comp_1(u) + Comp_0(u)} * PerUserGap(u)
\end{split}
\end{align*}

\vspace{5mm}
Finally we have:

\begin{align*}
\begin{split}
\mathop{{}\mathbb{E}}_i [Exposure(i, u) | Rel(i,u) = 1, S = 1] = & \ [Exposure(i, u) | Rel(i,u) = 1] + \frac{Comp_0(u)}{Comp_1(u) + Comp_0(u)} * PerUserGap(u)
\end{split}
\end{align*}

This equation is relatively intuitive. If there is no per-context bias, the relevant content from the sensitive group will get the same exposure. Otherwise, the exposure might be higher or lower depending on the bias.

We can then write the exposure of a given group as a function of the per-user bias, the performance and the competition.

\begin{align*}
\begin{split}
\mathop{{}\mathbb{E}}_i [Exposure(i, u) | Rel(i,u) = 1, S = 1]
= & \ \frac{ExposureAvailable}{|I|} \\
&+ (1 - \frac{Comp_0(u) + Comp_1(u)}{|I|}) * Performance(u)\\
&+ \frac{Comp_0(u)}{Comp_1(u) + Comp_0(u)} * PerUserGap(u)
\end{split}
\end{align*}

\vspace{50mm}
\textbf{Step 3: Gap aggregated across users.}

\begin{align*}
\begin{split}
AggregateGap =&  \ \frac{1}{\sum_{u} Comp_1(u)} \sum_{u} Comp_1(u) * \mathop{{}\mathbb{E}}_i [Exposure(i, u) | Rel(i,u) = 1, S = 1] \\
&- \frac{1}{\sum_{u} Comp_0(u)} \sum_{u} Comp_0(u) * \mathop{{}\mathbb{E}}_i [Exposure(i, u) | Rel(i,u) = 1, S = 0] \\
%
AggregateGap =& \ \sum_{u} \Biggl(\\
\frac{Comp_1(u)}{\sum_{u} Comp_1(u)} * \biggl(
&\frac{ExposureAvailable}{|I|}
+ (1 - \frac{Comp_0(u) + Comp_1(u)}{|I|}) * Performance(u)
+ \frac{Comp_0(u)}{Comp_1(u) + Comp_0(u)} * PerUserGap(u)\biggr) \\
- \frac{Comp_0(u)}{\sum_{u} Comp_0(u)} * \biggl(
&\frac{ExposureAvailable}{|I|}
+ (1 - \frac{Comp_0(u) + Comp_1(u)}{|I|}) * Performance(u)
- \frac{Comp_1(u)}{Comp_1(u) + Comp_0(u)} * PerUserGap(u)\biggr) \Biggr) \\
AggregateGap =& \ \sum_{u} \Biggl(\\
&\frac{ExposureAvailable}{|I|} * \biggl(\frac{Comp_1(u)}{\sum_{u} Comp_1(u)} - \frac{Comp_0(u)}{\sum_{u} Comp_0(u)}  \biggr) \\
& + Performance(u) * \biggl(\frac{Comp_1(u)}{\sum_{u} Comp_1(u)} - \frac{Comp_0(u)}{\sum_{u} Comp_0(u)} \biggr) * \biggl(1 - \frac{Comp_0 (u) + Comp_1 (u)}{|I|} \biggr) \\
& + PerUserGap(u) * \biggl(\frac{Comp_1(u)}{\sum_{u} Comp_1(u)}  * \frac{Comp_0(u)}{Comp_1(u) + Comp_0(u)} + \frac{Comp_0(u)}{\sum_{u} Comp_0(u)} * \frac{Comp_1(u)}{Comp_1(u) + Comp_0(u)} \biggr) \Biggr) \\
AggregateGap =& \ \sum_{u} \Biggl(\\
& + Performance(u) * \biggl(\frac{ Comp_1(u)}{\sum_{u} Comp_1(u)} - \frac{Comp_0(u)}{\sum_{u} Comp_0(u)}\biggr) * \biggl(1 - \frac{Comp_0 (u) + Comp_1 (u)}{|I|} \biggr) \\
& + PerUserGap(u) * \biggl(\frac{1}{\sum_{u} Comp_1(u)} + \frac{1}{\sum_{u} Comp_0(u)} \biggr) * \frac{Comp_0(u) * Comp_1(u)}{Comp_1(u) + Comp_0(u)} \Biggr) \\
\end{split}
\end{align*}

\end{document}